\documentclass[runningheads,11pt]{llncs}

\usepackage{graphicx}
\usepackage{color}
\usepackage{mathtools}

\usepackage{lineno}
\usepackage{amsfonts}

\usepackage{a4wide}

\newcommand{\cw}{{\mathtt{cw}}}
\newcommand{\tw}{{\mathtt{tw}}}

\newcommand{\nd}{{\mathtt{nd}}}
\newcommand{\tc}{{\mathtt{tc}}}

\newcommand{\mw}{{\mathtt{mw}}}
\newcommand{\fvs}{{\mathtt{fvs}}}
\newcommand{\fes}{{\mathtt{fes}}}
\newcommand{\ml}{{\mathtt{ml}}}
\newcommand{\mcs}{\omega}
\newcommand{\bw}{{\mathtt{bw}}}
\newcommand{\vc}{{\mathtt{vc}}}

\newcommand{\dist}{{\rm dist}}

\usepackage{etoolbox}
\newcommand*\linenomathpatchAMS[1]{%
  \expandafter\pretocmd\csname #1\endcsname {\linenomathAMS}{}{}%
  \expandafter\pretocmd\csname #1*\endcsname{\linenomathAMS}{}{}%
  \expandafter\apptocmd\csname end#1\endcsname {\endlinenomath}{}{}%
  \expandafter\apptocmd\csname end#1*\endcsname{\endlinenomath}{}{}%
}
\expandafter\ifx\linenomath\linenomathWithnumbers
  \let\linenomathAMS\linenomathWithnumbers
  \patchcmd\linenomathAMS{\advance\postdisplaypenalty\linenopenalty}{}{}{}
\else
  \let\linenomathAMS\linenomathNonumbers
\fi
\linenomathpatchAMS{gather}
\linenomathpatchAMS{multline}
\linenomathpatchAMS{align}
\linenomathpatchAMS{alignat}
\linenomathpatchAMS{flalign}

\begin{document}
\title{Computing $L(p,1)$-Labeling  
with Combined Parameters\thanks{This work is partially supported by JSPS KAKENHI Grant Numbers JP17K19960, JP17H01698, JP19K21537.}}
%
%
\author{Tesshu Hanaka\inst{1}\orcidID{0000-0001-6943-856X} \and
Kazuma Kawai\inst{2} \and
\\Hirotaka Ono\inst{2}\orcidID{0000-0003-0845-3947}}
\authorrunning{T. Hanaka et al.}
%
\institute{Chuo University, Tokyo 112-8551, Japan \\
\email{hanaka.91t@g.chuo-u.ac.jp}
\and
Nagoya University, Nagoya 464-8601, Japan\\
\email{\{kawai.kazuma@g.mbox.nagoya-u.ac.jp, ono@i.nagoya-u.ac.jp\}}}
\maketitle              
\begin{abstract}
Given a graph, an $L(p,1)$-labeling of the graph is an assignment $f$ from the vertex set to the set of nonnegative integers such that for any pair of vertices $(u,v),|f (u) - f (v)| \ge p$ if $u$ and $v$ are adjacent, and $f(u) \neq f(v)$ if $u$ and $v$ are at distance $2$. The \textsc{$L(p,1)$-labeling} problem is to minimize the span of $f$ (i.e.,$\max_{u\in V}(f(u)) - \min_{u\in V}(f(u))+1$). 
It is known to be NP-hard even for graphs of maximum degree $3$ or graphs with tree-width 2, whereas it is fixed-parameter tractable with respect to vertex cover number. Since vertex cover number is a kind of the strongest parameter, there is a large gap between tractability and intractability from the viewpoint of parameterization.
To fill up the gap, in this paper, we propose new fixed-parameter algorithms for \textsc{$L(p,1)$-Labeling} by the twin cover number plus the maximum clique size and by the tree-width plus the maximum degree. These algorithms reduce the gap 
in terms of several combinations of parameters. 

\end{abstract}
%

\noindent



\section{Introduction}

Let $G$ be an undirected graph, and $p$ and $q$ be constant positive integers. An $L(p, q)$-labeling of a graph $G$ is an assignment 
$f$ from the vertex set $V(G)$ to the set of nonnegative integers such that $|f (x)- f(y)| \ge p$ if $x$ and $y$ are adjacent and $|f(x)- f (y)| \ge  q$ if $x$ and $y$ are at distance $2$,  for all $x$ and $y$ in $V (G)$. 
We call the former \emph{distance-$1$ condition} and the latter \emph{distance-$2$ condition}.  A $k$-$L(p, q)$-labeling is an $L(p, q)$-labeling  $f: V (G)\to \{0, \ldots , k\}$, where the labels start from 0 for conventional reasons.
The $k$-$L(p, q)$-\textsc{Labeling} problem determines whether given $G$ has a 
$k$-$L(p, q)$-labeling, or not, and 
the $L(p, q)$-\textsc{Labeling} problem asks the minimum $k$ among all possible assignments.
The minimum value $k$ is called the $L(p, q)$-labeling number, and we 
denote it by $\lambda_{p,q}(G)$, or simply $\lambda_{p,q}$. Notice that we can use $k+1$ different labels when $\lambda_{p,q}(G) = k$. 

The original notion of $L(p, q)$-labeling can be seen in the context of frequency
assignment. Suppose that vertices in a graph represent wireless devices. The presence/absence of edges indicates the presence/absence of direct communication between the devices. If two devices are very close, that is, they are connected in the graph, they need to use sufficiently different frequencies, that is, their frequencies should be  apart at least $p$. If two devices are not very but still close, that is, they are at distance $2$ in the graph, their frequencies should be apart at least $q \ (\le p)$. 
Thus, the setting of $q=1$ as one unit and $p\ge q=1$ is considered natural and interesting, and the minimization of used range becomes the issue. Note that $L(1,1)$-labeling on $G$ is equivalent to the ordinary coloring on the square of $G$, which is denoted by $G^2$. 
From these, \textsc{$L(p,1)$-Labeling} for $p>1$ is intensively and extensively studied among several possible settings of $p$. 
In particular, \textsc{$L(2,1)$-Labeling} is considered the most important. A reason is that it is natural and suitable as a basic step to consider, and another reason is that the computational complexity (that is, hardness or polynomial-time solvability) tends to be inherited from $L(2,1)$ to $L(p,1)$ of $p>2$; for example, if \textsc{$L(2,1)$-Labeling} is NP-hard in a setting, the hardness proof could be modified to \textsc{$L(p,1)$-Labeling} in the same setting. Designing a polynomial time algorithm is also. 
We can find various related results on $L(p, q)$-labelings in comprehensive surveys by Calamoneri~\cite{calamoneri2011h} and by Yeh~\cite{yeh2006}. 

The notion of \textsc{$L(p,q)$-Labeling} firstly appeared in~\cite{Hale1980} and~\cite{Roberts1991}. 
Griggs and Yeh formally introduced the \textsc{$L(p,q)$-Labeling} problem (actually, it was \textsc{$L(2,1)$-Labeling})~\cite{griggs1992labelling}.
They also show that \textsc{$L(2,1)$-Labeling} is NP-hard in general.
Furthermore, \textsc{$L(2,1)$-Labeling} is shown to be NP-hard even for planar graphs, bipartite graphs, chordal graphs~\cite{bodlaender2004approximations}, graphs with diameter of $2$~\cite{griggs1992labelling} and graphs with tree-width $2$~\cite{fiala2005distance}. 
Moreover, for every $k\ge 4$, \textsc{$k$-$L(2,1)$-Labeling}, that is the decision version of \textsc{$L(2,1)$-Labeling} is NP-complete for general graphs~\cite{FKK2001} and even for planar graphs~\cite{EHN10}. These results imply that \textsc{$k$-$L(2,1)$-Labeling} is NP-complete
for every $\Delta\ge 3$, where $\Delta$ denotes the maximum degree. 
On the other hand, \textsc{$L(2,1)$-Labeling} can be solved in polynomial time for paths, cycles, wheels~\cite{griggs1992labelling}, but these are rather trivial. 
For non-trivial graph classes, only a few graph classes (e.g., co-graphs~\cite{chang19962} and outerplanar graphs~\cite{koller2004frequency})  are known to be solvable in polynomial time. 
In particular, Griggs and Yeh conjectured that \textsc{$L(2,1)$-Labeling} on trees was NP-hard, 
which was later disproved (under P$\neq$NP) by the existence of an $O(n^{5.5})$-time algorithm~\cite{chang19962}. It is now known that \textsc{$L(p,1)$-Labeling} on trees can be solved in linear time~\cite{HIOU2013}. For more algorithmic results, see~\cite{IJNC85}. 

From these results, we roughly understand the boundary between polynomial-time solvability and NP-hardness concerning graph classes, and studies are going to fixed-parameter (in)tractability. For a problem $A$ with input size $n$ and parameter $t$, $A$ is called  \emph{fixed-parameter tractable} with respect to $t$ if there is an algorithm whose running time is $g(t)n^{O(1)}$, where $g$ is a certain function. Such an algorithm is called a \emph{fixed-parameter} algorithm. If problem $A$ is NP-hard for a constant value of $t$, there is no fixed-parameter algorithm unless P$=$NP; we say $A$ is paraNP-hard.  
Unfortunately, \textsc{$L(2,1)$-Labeling} is already shown to be paraNP-hard for several parameters such as $\lambda_{2,1}$, maximum degree and tree-width as seen above. For positive results, there are fixed-parameter algorithms with respect to vertex cover number~\cite{FGK2011} or neighborhood diversity~\cite{fiala2018parameterized}. Note that vertex cover number is a  stronger parameter than tree-width, which means that if the vertex cover number is bounded, the tree-width is also.  There is still a gap on fixed-parameter (in)tractability between vertex cover number and  tree-width. For such a situation, two approaches can be taken. One is to finely classify intermediate parameters and see fixed-parameter (in)tractability for them, and the other is to combine two or more parameters and see fixed-parameter (in)tractability under the combinations. 
In this paper, we take the latter approach. 



\subsection{Our contribution}
In this paper, we present algorithms with combined parameters. The parameters that we focus on are clique-width ($\cw$), tree-width ($\tw$), maximum clique size ($\omega$), maximum degree ($\Delta$) and twin cover number ($\tc$). These are selected in connection with aforementioned parameters, $\lambda_{p,1}$, maximum degree and tree-width. Maximum clique size and clique-width are well used parameters weaker than tree-width. Maximum degree itself is a considered parameter, which is strongly related to $\lambda_{p,q}(G)$. In fact, it is easy to see that $\lambda_{p,1}\ge \Delta+p-1$, and $\lambda_{p,1}\le \Delta^2+(p-1)\Delta-2$~\cite{Goncalves2008}. Thus, $\lambda_{p,1}$ and $\Delta$ are parameters equivalent in terms of fixed-parameter (in)tractability. Twin cover number is picked up as a parameter that is moderately weaker than vertex cover number but stronger than clique-width and is also incomparable to neighborhood diversity. 

These parameters are ordered in the following two ways: (1) $(\vc \succeq) \{\tw,\tc \}\succeq \cw$ and (2) $(\lambda_{p,1} \simeq)\Delta \succeq \omega$. Here, for graph parameters $\alpha$ and $\beta$, $\alpha \succeq \beta$ represents that there is a positive function $g$ such that $g(\alpha(G))\ge \beta(G)$ holds for any $G$, and we denote $\alpha \simeq \beta$ if $\alpha \succeq \beta$ and $\beta \succeq \alpha$. For combined parameters of one from (1) and another from (2), we design fixed-parameter algorithms. Note that some combination yields essentially one parameter. For example, $\tw+\omega$ is equivalent to $\tw$, because $\tw \ge \omega-1$ holds. The obtained results are listed below: 
\begin{itemize}
    \item \textsc{$L(p,1)$-Labeling} can be solved in time $\Delta^{O(\tw \Delta)}n$ for $p\ge 1$. 
    Since it is known that $\tw \le 3 \cw \Delta -1$ (\cite{WG00}), it is also a $\Delta^{O(\cw \Delta^2)}n$-time algorithm, 
    which implies \textsc{$L(p,1)$-Labeling} is actually FPT with respect to $\cw+\Delta$. 
    This result also implies that \textsc{$L(p,1)$-Labeling} is FPT when parameterized by band-width.
    \item \textsc{$L(p,1)$-Labeling} is FPT when parameterized by $\tc+\omega$. Since $\tc+\omega \le \vc+1$ for any graph, it generalizes the fixed-parameter tractability with respect to vertex cover number in \cite{FGK2011}. 
    Since $\tc+\omega\ge \tw$, $\tc+\omega$ is located between $\tw$ and $\vc$. 
    \item  
    \textsc{$L(1,1)$-Labeling} is FPT when parameterized by \emph{only} twin cover number. This also yields a fixed-parameter $p$-approximation algorithm for \textsc{$L(p,1)$-Labeling} with respect to twin cover number.
\end{itemize}
Figure \ref{Fig:L21:result} illustrates the detailed relationship between graph parameters and the parameterized complexity of \textsc{$L(p,1)$-Labeling}.



\begin{figure}[tbp]
\centering
  \includegraphics[width=9cm]{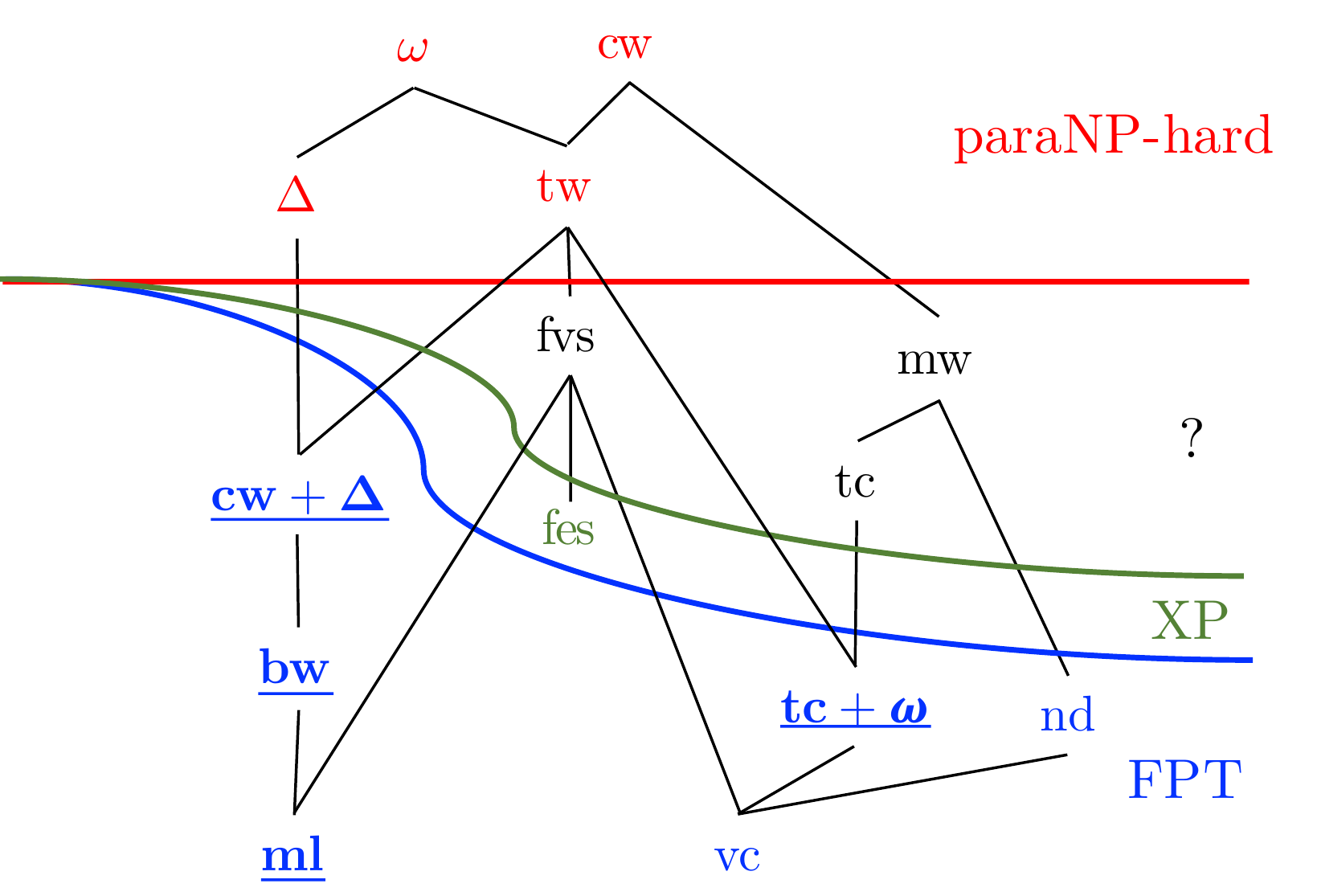}
\caption{The relationship between graph parameters and the parameterized complexity of \textsc{$L(p,1)$-Labeling}. Let $\mcs, \Delta, \cw, \mw, \nd, \tc, \tw, \fvs, \fes, \bw, \ml$, and $\vc$ denote maximum clique size, maximum degree, clique-width, modular-width, neighborhood diversity,  twin cover number, tree-width, feedback vertex set number, feedback edge set number, band-width, max leaf number, and vertex cover number, respectively.
Connections between two parameters imply that the upper is bounded by a function of the lower.
The underlines for parameters indicate that they are obtained in this paper. 
}
\label{Fig:L21:result}
\end{figure}

\subsection{Related work}
In this subsection, we mainly see related work on the parameterized complexity of \textsc{$L(p,1)$-Labeling}.  

We first see the case of $p>1$. It is NP-hard even on graphs of tree-width 2~\cite{fiala2005distance}.
Using stronger parameters than tree-width, Fiala et al. showed that \textsc{$L(p,1)$-Labeling} is fixed-parameter tractable when parameterized by vertex cover~\cite{FGK2011} and neighborhood diversity \cite{fiala2018parameterized}. 
Moreover, Fiala, Kloks and Kratochv{\'\i}l showed that the problem is XP when parameterized by feedback edge set number~\cite{FKK2001}.
For approximation, it is NP-hard to approximate \textsc{$L(p,1)$-Labeling} within a factor of $n^{0.5-\varepsilon}$ for any $\varepsilon>0$, 
whereas it can be approximated within $O(n(\log\log n)^2/\log^3 n)$~\cite{halldorsson2006approximating}.  

For \textsc{$L(1,1)$-Labeling}, it can be solved in time $O(\Delta^{2^{8(\tw+1)+1}} n+n^3)$, and hence it is XP by tree-width \cite{ZKN1996}. This result is tight in the sense of fixed-parameter (in)tractability, because  it is W[1]-hard when parameterized by tree-width \cite{FGK2011}.
Moreover, it can be solved in time $O(\cw^3 2^{6\cw}n^{2^{4\cw}+2^{2\cw}+1})$ \cite{Todinca2003}. 

Apart from \textsc{$L(p,1)$-Labeling}, twin cover number is a relatively new graph parameter, which is introduced in \cite{Ganian2015} as a stronger parameter than vertex cover number. 
In the same paper, many problems are shown to be FPT when parameterized by twin cover number, and it is getting to be a standard parameter (e.g., ~\cite{AEHLT2020,Bodlaender2020,EHKK2019,Gaspers2019,Jansen2019,Knop2019}). Recently, for \textsc{Imbalance}, which is one of graph layout problems, a parameterized algorithm is presented~\cite{MM2020}. 
It is interesting that they also adopt twin cover number plus maximum clique size as the parameters.  
\section{Preliminaries}\label{sec:pre}
In this paper, we use the standard graph notations.
Suppose that $G=(V,E)$ is a simple and connected graph with the vertex set $V$ and the edge set $E$. We sometimes use $V(G)$ or $E(G)$ instead of $V$ or $E$ respecively, to specify graph $G$. 
For $G=(V,E)$, we define the number of vertices $n=|V|$ and the number of edges $m=|E|$. 
For $V'\subseteq V$, we denote by $G[V']$ the subgraph of $G$ induced by $V'$. 
For two vertices $u$ and $v$, the \emph{distance} $\dist_G(u,v)$ is defined by the length of a shortest path between $u$ and $v$ where the length of a path is the number of edges of it.
We denote the closed neighbourhood and the open neighbourhood of a vertex $v$ by $N_G[v]$ and $N_G(v)$, respectively.
We also define $N^{\ell}_G[v]=\{u \mid \dist_G(u,v)=\ell\}$, $N^{\le \ell}_G[v]=\{u \mid \dist_G(u,v)\le \ell\}$, $N^{\ell}_G(v)=N^{\ell}_G[v]\setminus \{v\}$, and $N^{\le \ell}_G(v)=N^{\le \ell}_G[v]\setminus \{v\}$.
For a set $S\subseteq V$, let $N_G(S)=\bigcup_{v\in S}N_G(v)$ and $N_G[S]=\bigcup_{v\in S}N_G[v]$.
The degree of $v$ is denoted by $d_G(v)=|N_G(v)|$. 
The maximum degree of $G$ is denoted by $\Delta(G)$. 
For simplicity, we sometimes omit the subscript $G$.

The $k$-th power $G^k=(V,E^k)$ of a graph $G=(V,E)$ is a graph such that the set of vertices is $V$ and there is an edge $(u,v)$ in $E^k$ if and only if there is a path of length at most $k$ between $u$ and $v$ in $G$~\cite{Adrian2008}. 
In particular, $G^2$ is called the \emph{square} of $G$.

\subsection{Graph parameters}
\paragraph*{Clique-width} 
\begin{definition}\label{def:clique-width}
Let $c$ be a positive integer.
A \emph{$c$-graph} is a graph such that each vertex is labeled by an integer in $\{1, 2, \ldots, c\}$.
A vertex labeled by $i$ is called an \emph{$i$-labeled vertex}.
The \emph{clique-width} $\cw(G)$ is the minimum integer $c$ such that $G$ can be constructed by the following operations.
\begin{description}
    \item[{[O1]}] Add a new vertex with label $i \in \{1, 2, \ldots, c\}$;
    \item[{[O2]}] Take a disjoint union of $c$-graphs $G_1$ and $G_2$;
    \item[{[O3]}] Take two labels $i$ and $j$ and add an edge between every pair of an $i$-labeled vertex and a $j$-labeled vertex; 
    \item[{[O4]}] Relabel $i$-labeled vertices to label $j$.
\end{description}
\end{definition}

\paragraph*{Tree-width}
\begin{definition}[Tree Decomposition]
\label{def:treewidth}
A {\em tree decomposition} of a graph $G=(V,E)$ is defined as 
a pair $\langle {\cal X}, T\rangle$, where $T$ is a tree with node set $I(T)$ and ${\cal
X}=\{X_i \mid i\in I(T) \}$ is a collection of subsets, called {\em bags}, of $V$ such that:
\begin{enumerate}
\item \emph{(vertex condition)} $\bigcup_{i\in I(T)} X_i =V$;
\item \emph{(edge condition)} For every $\{u, v\}\in E$, there exists an $i\in I(T)$ such that 
	   $\{u, v\} \subseteq X_i$;
\item \emph{(coherence property)} 
For every $u\in V$, $I_u = \{i \in I(T) \mid u\in X_i \}$ induces a connected subtree of $T$.  
   
\end{enumerate}
The {\em width} of a tree decomposition is defined as $\max_{i\in I} |X_i| - 1$ and the {\em tree-width} of $G$, denoted by $\tw(G)$,  is defined as  the minimum width among all possible tree decompositions of $G$.

\end{definition}

\begin{definition}[Nice Tree Decomposition]\label{def:nicetree}
 A tree decomposition  $\langle {\cal X}, T\rangle$ is called a {\em nice tree decomposition} if it satisfies the following:
\begin{description}
\item[1.] $T$ is rooted at a designated node $r(T) \in I$ satisfying $X_{r(T)}=\emptyset$, called 
the {\em root node}.
\item[2.] Every node of the tree $T$ has at most two children.
\item[3.] Each node $i$ in $T$ has one of the following five types:
\begin{itemize}
\item A {\em leaf} node $i$ has no children and its bag $X_i$ satisfies $X_i = \emptyset$,
\item An {\em introduce vertex} $v$ node $i$  has exactly one child $j$ with $X_i = X_j \cup \{v\} $ for a vertex $v\in V$,
\item An {\em introduce edge} $\{u,v\}$ node $i$ has exactly one child $j$ and  labeled with an edge $\{u, v\} \in E$ where $u, v \in X_i$ and  $X_i = X_j$,
\item A {\em forget} $v$  node $i$ has exactly one child $j$ and satisfies $X_i = X_j \setminus \{v\}$ for
a vertex $v \in V$, and
\item A {\em join} node $i$ has exactly two children $j_1, j_2$ and satisfies $X_{j_1}=X_i$ and $X_{j_2}=X_i$.
\end{itemize} 
\end{description}
We additionally require that every edge in $E$ is introduced exactly once.
\end{definition}
By the last statement, every edge is assigned to exactly one node. An assignment is done by an introduce edge node, for a pair of vertices that have already been introduced. This implies that for an introduce vertex $v$ node $i$, $v$ is an isolated vertex in $G_i$, where $G_i=(V_i, E_i)$ is defined by $V_i$, the union of all bags $X_j$ such that $j=i$ or $j$ is a descendant of $i$, and $E_i\subseteq E$, the set of all edges introduced at $i$ (if $i$ is an introduce edge node) or a descendant of $i$.     

One can compute the treewidth of $G$ and its tree decomposition in time ${\tw}^{O({\tw}^3)}n$~\cite{bodlaender1996linear}.
Moreover, any tree decomposition with $\ell$ nodes can be transformed  to a nice tree decomposition with $O({\tw}\cdot n)$ bags
and the same width in time $O(\tw\cdot \max\{\ell,n\})$~\cite{Cygan2015}.

\paragraph*{Twin cover}
Two vertices $u,v$ are called {\em twins} if both $u$ and $v$ have the same neighbors. 
Moreover, if twins $u,v$ have edge $\{u,v\}$, they are called {\em true twins} and the edge is called a {\em twin edge}. 
Then a {\em twin cover} of $G$ is defined as follows.
\begin{definition}[\cite{Ganian2015}]\label{def:twincover}
A set of vertices $X$ is a {\em twin cover} of $G$ if every edge $\{u,v\}\in E$ satisfies either
\begin{itemize}
\item $u\in X$ or $v\in X$, or
\item $u,v$ are true twins.
\end{itemize}  
The {\em twin cover number} of $G$, denoted by $\tc(G)$, is defined as the minimum size of twin covers in $G$. 
\end{definition}
An important observation is that the complement $V\setminus X$ of a twin cover $X$ induces disjoint cliques.
Moreover, for each clique $Z$ of $G[V \setminus X]$, $N(u) \cap X = N(v) \cap X$ for every $u, v \in Z$~\cite{Ganian2015}.

A {\em vertex cover} $X$ is the set of vertices such that for every edge, at least one endpoint is in $X$.
The {\em vertex cover number} of $G$, denoted by $\vc(G)$, is defined as the minimum size of vertex covers in $G$. Since every vertex cover of $G$ is also a twin cover of $G$, $\tc(G)\le \vc(G)$ holds.
Also, for any graph $G$, we have $\tc(G)+\omega(G)\le \vc(G)+1$.

\paragraph*{Band-width}
For a graph $G=(V,E)$, the \emph{band-width} $\bw(f)$ of a map $f: V\rightarrow [1,n]$ is defined by $\max_{(i,j)\in E}|f(i)-f(j)|$. The \emph{band-width} $\bw(G)$ of $G$ is defined by the minimum value of $\max_{(i,j)\in E}|f(i)-f(j)$ among all possible $f$, that is, $\bw(G)=\min_{f: V\rightarrow [1,n]}\bw(f)$.

\subsection{Integer Linear Programming}
\textsc{Integer Linear Programming Feasibility} is formulated as follows.

\begin{description}
\item[Input: ] An $q\times p$ matrix $A$ with integer elements, an integer vector $b\in \mathbb{Z}^q$
\item[Question: ] Is there a vector $x\in \mathbb{Z}^p$ such that $A\cdot x \le b$.
\end{description}

Lenstra \cite{Lenstra1983} proved that \textsc{Integer Linear Programming Feasibility} is FPT when parameterized by the number of variables and the running time was improved by Frank and Tardos \cite{Frank1987} and by Kannan \cite{Kannan1987}.
\begin{theorem}[\cite{Lenstra1983,Frank1987,Kannan1987}]\label{thm:ILP}
\textsc{Integer Linear Programming Feasibility}  can be solved using $O(p^{2.5p+o(p)}\cdot L)$ arithmetic operations and space polynomial in $L$, where $L$ is the number of bits in the input.
\end{theorem}
\section{Parameterization by $\cw + \Delta$ and $\tw + \Delta$}\label{sec:treewidth}
As \textsc{$L(p,1)$-Labeling} is paraNP-hard for tree-width, so is for clique-width.  
In this section, as a complement, we show that \textsc{$L(p,1)$-Labeling} (actually, \textsc{$L(p,q)$-Labeling} for any constant $p$ and $q$) is fixed-parameter tractable when parameterized by $\cw + \Delta$.

To this end, we give a fixed-parameter algorithm for \textsc{$L(p,1)$-Labeling} parameterized by not $\cw+\Delta$ but 
$\tw+\Delta$, which actually implies that the problem is FPT with respect to $\cw+\Delta$, because it is known that $\tw \le 3\cw-1$~\cite{WG00}.   
The running time of the algorithm is $\Delta^{O(\tw \Delta)}n$, and so it is $\Delta^{O(\cw \Delta^2)}n$. 

In the algorithm, we first construct the square $G^2$ of $G$ and then compute \textsc{$L(p,1)$-Labeling} of $G$  by dynamic programming on a nice tree decomposition $\langle {\mathcal X'}, T'\rangle$ of $G^2$. Actually, the algorithm runs for \textsc{$L(p,q)$-Labeling} though the running time depends on $\lambda$. 
One can obtain the square of $G^2$ in time $O(m\Delta(G))=O(\Delta(G)^2n)$.
We then prove the following lemma.
\begin{lemma}\label{lem:G^2}
Given a tree decomposition of a graph $G$ of width $t$ with $\ell$ bags, one can construct a tree decomposition of $G^2$ of width at most $(t+1)\Delta(G)+t$ with $\ell$ bags in time $O(t\Delta(G)\ell)$.
\end{lemma}
\begin{proof}


We are given a tree decomposition $\langle {\mathcal X}, T\rangle$ of $G$ of width $t$.
Let $X'_i=X_i \cup N(X_i)$ and ${\mathcal X'}=\{X'_i\mid i\in I(T)\}$ be the set of bags. We here define $\langle {\mathcal X'}, T'\rangle$ as a tree decomposition of $G^2$, 
where $T'$ and $T$ are identical; $T$ and $T'$ has the same node set and the same structure, where each $i \in I(T')$ corresponds to $i\in I(T)$. 
In the following, we denote $\langle {\mathcal X'}, T\rangle$ instead of $\langle {\mathcal X'}, T'\rangle$. 

We can see that $\langle {\mathcal X'}, T\rangle$ is really a tree decomposition of $G^2$ with width $(t+1)\Delta(G)+t$. It satisfies the properties of tree decomposition indeed: Since $\bigcup_{i\in I} X'_i = \bigcup_{i\in I} (X_i \cup N(X_i)) = V(G) = V(G^2)$, the vertex condition is satisfied. 
We next see edge condition. For each $e\in E$, there is $X_i$ containing $e$, so $e\in X'_i$. For each $\{u,v\}\in E^2 \setminus E$, 
there is a vertex $v' (\neq u,v)$ such that $\{u,v'\}\in E$ and $\{v',v\}\in E$. Thus there is $X_i$ satisfying $\{u,v'\}\subseteq X_i$, which implies $\{u,v\}\subseteq X_i \cup \{v\} \subseteq X_i \cup N(\{v'\}) \subseteq X'_i$. These show that the edge condition is satisfied.  



Finally, 
we check coherent property: we show that for every $u\in V$, $I'_u = \{i \in I(T) \mid u\in X'_i \}$ induces a connected subtree of $T$. 
Note that 
\[
I'_u = \{i \in I(T) \mid u\in X'_i \}=\{i \in I(T) \mid u\in X_i \}\cup \bigcup_{v\in N(u)}\{i \in I(T) \mid v\in X_i \}. 
\]
Here, the subgraph $T_v$ of $T$ induced by $\{i \in I(T) \mid u\in X_i \}$ is connected by 
the coherent property of $\langle {\mathcal X}, T\rangle$. 
Also for each $v\in N(u)$, the subgraph $T_v$ of $T$ induced by $\{i \in I(T) \mid v\in X_i \}$ is connected. By $\{u,v\}\in E$, 
the edge condition of $\langle {\mathcal X}, T\rangle$ implies that there exists a bag $X_j$ containing both $u$ and $v$. Since $T_u$ and $T_v$ has a common node $j$, the subgraph of $T$ induced by $\{i \in I(T) \mid u\in X_i \} \cup \{i \in I(T) \mid v\in X_i \}$ is also connected, 
which leads that the subgraph of $T$ induced by $I'_u$ is also connected.

Hence, $\langle {\mathcal X'}, T\rangle$ is a tree decomposition of $G^2$. 
Since the size of bag $X'_i$ is $|X'_i|=|X_i \cup N(X_i)|=|\bigcup_{u\in X_i}N[u]|\le (t+1)(\Delta(G)+1)$, 
the width is at most $(t+1)(\Delta(G)+1)-1=(t+1)\Delta(G)+t$. 
The construction of $\langle {\mathcal X'}, T\rangle$ is done by preparing each $X'_i$, which takes $O(t\Delta(G))$ steps for each $i$. 
Thus it can be done in time $O(t\Delta(G)\ell)$ in total. \qed
\end{proof}

\begin{corollary}
$\tw(G^2)\le \tw(G)(\Delta(G)+1)-1$ holds.
\end{corollary}
Note that this bound is tight, because $\tw(K_{1,\Delta})=1$ and $\tw(K_{1,\Delta}^2)=\tw(K_{\Delta+1})$ $=\Delta$. 
By the above lemma, the tree-width of $G^2$ is bounded if $\tw(G)$ and $\Delta(G)$ are bounded. 
Thus we can design a dynamic programming algorithm on a nice tree decomposition of $G^2$, although we omit the detail.

\begin{lemma}\label{lem:DP:tw}
Given a nice tree decomposition of $G^2$ of width at most $t$, one can compute \textsc{$k$-$L(p,q)$-Labeling} on $G$ in time $O((k+1)^{t+1} t^2n)$.
\end{lemma}

Here, one can construct a tree decomposition $\langle {\mathcal X}, T\rangle$ of $G$ of width $5\tw(G)+4$ with $O(n)$ bags in time $2^{O(\tw(G))}n$ \cite{Bodlaender2016}. 
By Lemma \ref{lem:G^2}, we can obtain a tree decomposition $\langle {\mathcal X}', T\rangle$ of $G^2$ of width $(5\tw(G)+4+1)\Delta(G) + 5\tw(G)+4=O(\tw(G)\Delta(G))$ from $\langle {\mathcal X}, T\rangle$ in time $O(\tw(G)\Delta(G)n)$. By Lemma \ref{lem:DP:tw} and $\lambda_{p,q}\le \max\{p,q\}\Delta^2$, we have the following theorem .


\begin{theorem}\label{thm:tw}
For any positive constant $p$ and $q$, there is an algorithm to solve 
\textsc{$L(p,q)$-Labeling} in time $\Delta^{O(\tw \Delta)}n$, which is also bounded by $\Delta^{O(\cw \Delta^2)}n$.
\end{theorem}

Since $\tw(G)\le \bw(G)$ and $\Delta(G)\le 2\bw(G)$,
we have the following corollary.
\begin{corollary}
For any positive constant $p$ and $q$, \textsc{$L(p,q)$-Labeling}  is fixed-parameter tractable when parameterized by band-width.
\end{corollary}

\color{black}

\section{Parameterization by twin cover number}\label{sec:twin_cover}
\subsection{\textsc{$L(p,1)$-Labeling} parameterized by $\tc+\omega$}

In this section, we design a fixed-parameter algorithm for \textsc{$L(p,1)$-Labeling} with respect to $\tc + \omega$.
Notice that for a twin cover $X$ of $G=(V,E)$, each of the connected components of $G[V\setminus X]$ forms a clique.  
We categorize vertices in $V\setminus X$ with respect to the neighbors in $X$.
Let $T_1, T_2, \ldots, T_s$ be the sets of vertices having common neighbors in $X$, called \emph{types} of vertices in $V\setminus X$, where $s$ is the number of types. 
Moreover, we say that a clique $C\subseteq V\setminus X$ is of type $T_i$ if $C\subseteq T_i$.
Note that $V\setminus X=\bigcup^{s}_{i=1}T_i$.
Let $n_i=|T_i|$ and $\omega_i$ be the maximum clique size in $T_i$.

We first see a general property about cliques with a common neighbor:  
Suppose that a graph $G$ consists of only cliques and common neighbors $Y$ of all the vertices in the cliques. That is, all the vertices are within distance $2$. Also suppose that vertices in $Y$ has some labels $a_1, a_2, \ldots, a_{|Y|}$ and $L$ is a set of labels that are at least $p$ apart from $a_1, a_2, \ldots, a_{|Y|}$. 
Then the following lemma holds. 
\begin{lemma}\label{lem:maxclique}
Suppose that a graph $G$ and a label set $L$ are defined as above, and let $C_1, C_2, \ldots, C_h$ be the set of the cliques, in the descending order of the size. 
If $|L|\ge \sum_{j} |C_j|$ and $\sum_{j} |C_j| \ge p|C_1|$ hold, there exists an $L(p,1)$-labeling of $C_1,\ldots,C_h$ using only labels in $L$. 
\end{lemma}

\begin{proof}
Let $n'=\sum_{j} |C_j|$ and $\omega=|C_1|$. The statement of the lemma is rewritten as ``if $|L|\ge n'$ and $n'\ge p \omega$, all the cliques can be properly labeled with $L$''. Let us assume $L=\{l_1,l_2,\ldots,l_{n'}\}$. Since we can use distinct labels for vertices in $C_1, C_2,\ldots,$ $C_h$, only the distance-1 condition inside of a same clique matters. 
If $n' \equiv 1$ (mod $p$), we label the vertices in $C_{1},  C_{2}, \ldots$, $C_{n'}$ in this order by using labels in order of $l_1, l_{p+1}, l_{2p+1}, \ldots, l_{n'}, l_2, l_{p+2},l_{2p+2} \ldots$, $l_{n'-p+2}, l_{3}\ldots, l_p, l_{2p}, \ldots, l_{n'-1}$. Note that the vertices in $C_{1}$ are  labeled by $l_1, l_{p+1},$ $\ldots,$ $l_{p(\omega-1)+1}$ (note that $p\omega\le n'$). Since the difference between $l_{\alpha p+i}$ and $l_{(\alpha +1)p+i}$ for each $i$ and $\alpha$ is at least $p$, the labeling for cliques does not violate the distance-1 condition. We can choose similar orderings for the other residuals. 
\qed
\end{proof}

Now we go back to the algorithm parameterized by $\tc+\omega$. 
Given a twin cover $X$, we say that a $k$-$L(p,1)$-labeling is \emph{good} for $X$ if it uses only labels in $\{0,1,\ldots, (2p-1)|X|-p\} \cup \{k -(2p-1)|X|+p, \ldots, k \}$ for $X$. 
The following lemma is also important. 
\begin{lemma}
\label{lem Xlabel}
Let $X$ be a twin cover in $G$ such that each $T_i$ satisfies $\omega_i \le 
n_i/p$. 
Then any $k$-$L(p,1)$-labeling $f$ of $G$ can be transformed into a good $k$-$L(p,1)$-labeling $f^*$.
\end{lemma}

\begin{proof}
Let $f$ be an $L(p,1)$-labeling, and $a,b\in \{0, 1, \ldots, k\}$ be two labels such that (1) they are not used in $X$, (2) they are at least $p$ apart from all the labels used in $X$, and (3) there is at least one label $l$ used in $X$ where $a+p\le l \le b-p$. 
Then we rotate labels between $a$ and $b$ in $f$ as follows: 
$a \rightarrow a+1$, 
$a+1 \rightarrow a+2$, $\ldots$, $b \rightarrow a$. 


Let $f'$ be a labeling obtained by the above relabeling.
The rotation does not affect the distance-2 condition, though it may affect distance-1 condition.  
As for $X$, we notice that only labels in $\{a+p, \ldots, b-p\}$ are changed in $X$, which does not yield any new conflict inside of $X$. Therefore, $f'$ satisfies the distance-1 condition of $L(p,1)$-labeling in $G[X]$. Also $b-p+1$ is only a label that could be newly used in $X$ of $f'$, which does not affect any label in $V\setminus X$;  $f'$ also satisfies the distance-1 condition of $L(p,1)$-labeling between $X$ and $V\setminus X$.

We see that $f'$  does not violate the condition of $L(p,1)$-labeling within $X$ and between $X$ and $V\setminus X$. 
On the other hand, it may violate the condition within $V\setminus X$. For example, if a clique in $G[V\setminus X]$ has two vertices labeled with $b-p+1$ and $b+1$ in $f$, they are labeled with $b-p+2$ and $b+1$ in $f'$, which violates the distance-1 condition by $(b+1)-(b-p+2)=p-1$. 
Fortunately, such a violation can be easily avoided by further relabeling vertices in $V\setminus X$ as follows.

For each $T_i$, we first observe that labels used in $f$ for $T_i$ are different from each other due to the distance 2-condition, as so in $f'$. 
A problem may occur inside of a clique, which may violate the distance-1 condition. 
However, even if a conflict occurs, we can obtain a proper $k$-$L(p,1)$-labeling by relabeling the vertices in $T_i$ with the same label set. This is because 
the cliques inside of $T_i$ have exactly same neighbors and $p\omega_i \le n_i$ holds, by which we can apply the argument of Lemma \ref{lem:maxclique}.

The above procedure can push up a label in a middle range used in $X$.  
It can be applied as long as a triplet of $a, b$ and $l$ exists. By the definition of $a$ and $b$, $l$ can exist only when $b-a\ge 2p$. For example, the triplet of $(a, l, b)$ is possible for $l=a+p$ and $b=a+2p$, but we cannot take $l$ for $b< a+2p$. Consider the labeling where all the vertices in $X$ are labeled by $|X|$ labels near $k$: $k -(2p-1)|X|+p, k -(2p-1)(|X|-1)+p, k -(2p-1)(|X|-2)+p,\ldots, k -(2p-1)+p, k -p+1$. It is easy to see that we cannot take $a$ and $b$ for the labeling, though we can take $a$ and $b$  if we use $k -(2p-1)|X|+p-1$ or a smaller label instead of $k -(2p-1)|X|+p$. 
On the other hand, consider the labeling where all the vertices in $X$ are labeled by $|X|$ labels near $0$: $p-1, 3p-2, \ldots, (p-1)+(2p-1)(|X|-2), p-1+(2p-1)(|X|-1)(=(2p-1)|X|-p)$. We cannot take $a$ and $b$ again. 

By these, if we cannot apply the above procedure, all the labels for $X$ are in $\{0,1,\ldots,(2p-1)|X|-p-1,(2p-1)|X|-p\}\cup \{k -(2p-1)|X|+p,k -(2p-1)|X|+p+1,\ldots,k-1, k\}$. 
Hence, by applying the above procedure repeatedly, we eventually obtain a good $k$-$L(2,1)$-labeling $f^*$. \qed

\end{proof}

From Lemma \ref{lem Xlabel}, we immediately obtain the following corollary.
\begin{corollary}\label{cor:good_labeling}
Let $X$ be a twin cover in $G$ such that each $T_i$ satisfies $\omega_i \le 
n_i/p$. 
If there is a $k$-$L(p,1)$-labeling in $G$, then
there is a good $k$-$L(p,1)$-labeling for $X$ in $G$.
\end{corollary}
Therefore, we consider to find only a good $L(p,1)$-labeling. By using the corollary, we can show that \textsc{$L(p,1)$-Labeling} is fixed-parameter tractable when parameterized by $\tc+\omega$.

\begin{theorem}\label{thm:L-21:twin}
\textsc{$L(p,1)$-Labeling} is fixed-parameter tractable when parameterized by $\tc+\omega$.
\end{theorem}

\begin{proof}
We present an algorithm to solve $k$-\textsc{$L(p,1)$-Labeling} instead of \textsc{$L(p,1)$-Labeling}.  
We first compute a minimum twin cover $X$ in time $O(1.2738^{\tc} + \tc n + m)$ \cite{Ganian2015}.
For twin cover $X$, we define $T_i$'s. Then, we define another twin cover of 
$X'=X \cup \bigcup_{i:\omega_i > {n_i}/{p}}T_i$. Since $X$ is a twin cover, 
$X'$ is also. The size of $X'$ is bounded by $\tc + 2^{\tc}\cdot p\cdot \omega$, 
because the number of types is at most $2^{\tc}$ and the size of $T_i$ joining $X$  is at most $p\cdot \omega$, where $\omega$ is the maximum clique size.   
Let $\tc'=|X'|$.


\medskip

We are now ready to present the core of the algorithm. We classify an instance into two cases. If $k$ is small enough, we can apply a brute-force type algorithm. Otherwise, we try to find a good $k$-$L(p,1)$-labeling. 

\medskip

\noindent{\bf Case 1. $k < 8p\tc' $}

For each type $T_i$, the distance between two vertices in $T_i$ is at most $2$.
Thus, the labels of vertices in $T_i$ must be different each other.
Due to $k < 8p\tc' $, if $|T_i|\ge 8p\tc'$, we immediately conclude that the input is a no-instance.
Otherwise, $n = |X'|+\sum |T_i| \le \tc' + 8p\tc' 2^{\tc}$ holds, 
because the number of $T_i$'s is at most $2^{\tc}$. 
Thus we check all the possible labelings in time $O((8p\tc')^{8p\tc' 2^{\tc}})$. 

\medskip

\noindent{\bf Case 2. $k \ge 8p\tc'$}

Let $\mathcal{C}_{0}, \mathcal{C}_{1}, \ldots \mathcal{C}_{t}$ be the family of all possible set systems on $\{T_1, \ldots, T_s\}$ such that whenever $T_j$ and $T_{j'}$ are distinct elements of a system $\mathcal{C}_i$ then $N(T_j)\cap N(T_{j'})=\emptyset$.
We define $\mathcal{C}_0$ as an empty set. 
These are introduced to describe a set of $T_j$'s that can use a same label. 
For each $\mathcal{C}_i$, we prepare a set $L_{i}$ of labels, which will be used 
during the execution of the algorithm to represent the set of labels that 
could be used for vertices in $T_j \in \mathcal{C}_i$. Note that  
$L_{0}, L_{ 1}, \ldots ,L_{t}$ must be disjoint each other, 
and a label in $L_i$ is used exactly once per $T_j$. 
We also define $L_0$ as the set of labels that are not used in $V\setminus X'$.
Note that each $L_i$ can be empty. 

By Corollary \ref{cor:good_labeling}, there is a good $k$-$L(2,1)$-labeling for $X$ such that vertices in $X$ only use labels in $\{0,1,\ldots, 2p(\tc'-1) -p \} \cup \{k -2p(\tc'-1) +p, \ldots, k\}$ if the input is an yes-instance.
Thus we try all the possible partial labelings for $X$, each of which uses only labels in $\{0,1,\ldots, 2p(\tc'-1) -p \} \cup \{k -2p(\tc'-1) +p, \ldots, k\}$.
Since the number of labels is $2(2p(\tc'-1)-p+1)\le 4p\tc'$, 
there are at most $(4p\tc' )^{\tc'}$ possible labelings of $X$. 
For each of them we further try all the possible placement of labels in $\{0,1,\ldots, 2p(\tc'-1) -1\} \cup \{k -2p(\tc'-1) +1, \ldots, k\}$ into $L_{0}, L_{1}, \ldots, L_{t}$, which is a little wider than above. 
The number of possible placements is at most $t^{4p\tc'}$ due to the disjointness of $L_i$'s.
Therefore, the total possible nonisomorphic partial labelings is bounded by $(4p\tc' )^{\tc}\cdot t^{4p\tc'}$. It should be noted that no vertex will be labeled by a label in $\{0,1,\ldots, 2p(\tc'-1) -1\} \cup \{k -2p(\tc'-1) +1, \ldots, k\}$ hereafter.
Thus we consider how we use labels in $\{2p(\tc'-1), \ldots, k -2p(\tc'-1), \ldots, k\}$ for $V\setminus X$, which does not yield any conflict with $X$.  

We then formulate as Integer Linear Programming how many labels should be placed in $L_{0}, L_{1}, \ldots , L_{t}$ for one partial labeling using $\{0,1,\ldots, 2p(\tc'-1) -p\} \cup \{k -2p(\tc'-1) +p, \ldots, k\}$.
 For a fixed partial labeling, 
let $a_i$ be the number of labels that have been already assigned to $L_i$ there, and $x_i$ be a variable representing the number of labels used in $L_i$ in the desired labeling.

The following is the ILP formulation. 
\begin{align*}
  \begin{cases}
    x_0 + \cdots + x_t \le k + 1 & \\
    x_i \ge a_i,  & {\rm for}\ i \in \{0,\ldots, t\} \\
    \sum_{i:T_{j}\in \mathcal{C}_i}x_i = |T_j|,  & {\rm for}\ j \in \{1, \ldots, s\} \\
  \end{cases}
\end{align*}
The first constraint shows that the total number of labels is at most $k+1$. Note that the number of unused labels is $x_0$. The second one is for consistency to the partial labeling. The last one, which is the most important, guarantees that every vertex in $T_j$ can receive a label; the number of usable labels is $|\{i \mid T_j \in \mathcal{C}_i\}|$, because a label in $L_i$ is used exactly once per $T_j$. 

If the above ILP has a feasible solution, it is possible to assign labels to all the vertices in $V\setminus X$ if we ignore the distance-1 condition inside of each clique.   
Actually, we can see that the information is sufficient to give a proper $k$-$L(p,1)$-labeling. At the beginning of the algorithm, we take twin cover $X'$, which means that for every $T_i \subseteq V\setminus X$, $n_i \ge p \omega_i$ holds. Since cliques in $G[T_i]$ have common neighbors and $n_i \ge p \omega_i$, only the number of available labels matters by Lemma \ref{lem:maxclique}. 
Since the existence of an ILP solution guarantees this, we can decide whether a partial labeling can be extended to a proper $k$-$L(p,1)$-Labeling, or not.  

Because $s\le 2^{\tc}$ and $t\le 2^{2^{\tc}}$, the number of variables of the above ILP is at most
$2^{2^{\tc}}$; it can be solved in FPT time with respect to ${\tc}$ by Theorem \ref{thm:ILP}~\cite{Lenstra1983,Frank1987,Kannan1987}. 
Since $\tc'\le \tc + 2^{\tc}\cdot p \cdot \omega$, the total running time is FPT time 
with respect to $\tc + \omega$. \qed




\end{proof}



\color{black}

\subsection{\textsc{$L(1,1)$-Labeling} parameterized by twin cover number}
Unlike $L(p,1)$-labeling with $p\ge 2$, the distance-1 condition of $L(1,1)$-labeling requires just that the labels between adjacent vertices are different. Thus, \textsc{$L(1,1)$-Labeling} seems to be easier than \textsc{$L(p,1)$-Labeling} with $p\ge 2$.
Actually, we can show that \textsc{$L(1,1)$-Labeling} is fixed-parameter tractable parameterized only by twin cover number.

\begin{lemma}\label{lem:twin:label}
For a graph $G$, let $u$ and $v$ be twins with edge $\{u,v\}\in E(G)$.  Let $G'$ be the  graph of $G'=(V,E')$, where $E'=E(G) \setminus \{\{u,v\}\}$.  
Then any $L(1,1)$-labeling on $G'$ is also an $L(1,1)$-labeling on $G$ and verse visa. 
\end{lemma}
\begin{proof}
The statement is true, if $N_{G}^{\le 2}[y]=N_{G'}^{\le 2}[y]$ holds for any vertex $y \in V$, and we show this here. Since $N_{G}^{\le 2}[y] \supseteq N_{G'}^{\le 2}[y]$ is obvious, we show $N_{G}^{\le 2}[y] \subseteq N_{G'}^{\le 2}[y]$, that is, for any $y\in V$, if $w\in N_{G}^{\le 2}[y]$, $w$ also belongs to $N_{G'}^{\le 2}[y]$. Note that $w\in N_{G}^{\le 2}[y]$ means there is a path with length at most $2$ between $y$ and $w$. If $G$ has such a path between $y$ and $w$ not containing $\{u,v\}$, $G'$ also does. Thus, $w\in N_{G'}^{\le 2}[y]$. Otherwise, every path with length at most $2$ between $y$ and $w$ in $G$ contains 
$u$ and $v$, which implies that either $y$ or $w$ is $u$ or $w$. We just see the case when $y=u$ for symmetry, and take such a path between $y(=u)$ and $w$.  
If the path length is 1 (that is, $w=v$) in $G$, $y(=u)$ and $w(=v)$ has a common neighbor because $u$ and $v$ are twins, which implies that $w$ and $y$ are within distance $2$ in $G'$. If the path length is $2$, the path forms $(y,v,w)$. Namely, $w$ is a neighbor of $v$ and also of $u(=y)$ in $G'$. This completes the proof. \qed


\end{proof}

By Lemma \ref{cor:twin:label}, 
we immediately obtain the following corollary.
\begin{corollary}\label{cor:twin:label}
A minimum $L(1,1)$-labeling in $G'$ is a minimum $L(1,1)$-labeling in $G$.
\end{corollary}

Let $X$ be a twin cover again, and then each connected component in $G[V\setminus X]$ forms a clique, each of the edges in which are twin edges. Lemma \ref{cor:twin:label} implies that graph $G'$ obtained by removing all the edges in $G[V\setminus X]$ has the same $L(1,1)$-labeling number of $G$.   
The above deletion shows that $X$ is also a vertex cover of $G'$.  
Since \textsc{$L(1,1)$-Labeling} is fixed-parameter tractable when parameterized by vertex cover number \cite{FGK2011}, 
it is also fixed-parameter tractable when parameterized by twin cover number.

\begin{theorem}\label{thm:L-11:twin}
\textsc{$L(1,1)$-Labeling} is fixed-parameter tractable when parameterized by twin cover number.
\end{theorem}

In \cite{Georges1995GeneralizedVL}, it is shown that for $G$ and a positive integer $c$, $\lambda_{cp,cq}(G)=c\lambda_{p,q}(G)$ holds. Thus we have  $\lambda_{1,1}(G)\le \lambda_{p,1}(G)\le \lambda_{p,p}(G)=p\lambda_{1,1}(G)$, which gives an approximation for \textsc{$L(p,1)$-Labeling}. In fact, by replacing the labels of an optimal $L(1,1)$-labeling of $G$ with multiples of $p$, we obtain an $L(p,1)$-labeling whose factor is at most $p$.      

\begin{corollary}\label{cor:tc-approx}
For \textsc{$L(p,1)$-Labeling}, there is a fixed-parameter $p$-approximation algorithm with respect to twin cover number.  
\end{corollary}



\section{Concluding Remarks}\label{sec:conclusion}
In this paper, we studied the parameterized complexity of \textsc{$L(p,1)$-Labeling}. 
The parameterization is mainly by combination of two parameters, because the problem 
is known to be NP-hard even on graphs of tree-width 2. We show that it is FPT 
when parameterized by clique-width plus maximum degree and twin cover number plus maximum clique size. The former result implies \textsc{$L(p,1)$-Labeling} is FPT when parameterized by bandwidth, and the latter strengthens the fact that \textsc{$L(p,1)$-Labeling} is FPT when parameterized by vertex cover number~\cite{FGK2011}.
For \textsc{$L(1,1)$-Labeling}, we further prove that it is FPT with respect to only twin cover number.

Some FPT results hold for more general settings, that is, \textsc{$L(p,q)$-Labeling} with any constant $p$ and $q$. For example, \textsc{$L(p,q)$-Labeling} with any constant $p$ and $q$ is 
FPT when parameterized by clique-width plus maximum degree, or tree-width plus maximum degree. This implies that bounding maximum degree is essential for NP-hardness, because \textsc{$L(p,q)$-Labeling} for trees 
(i.e., graphs with tree-width 1) is NP-hard for every pair of $p$ and $q$ having no common divisor~\cite{FGK2008}. 

An interesting open question is whether \textsc{$L(p,1)$-Labeling}  parameterized by only twin cover number is FPT or not.

\subsection*{Acknowledgements}
We are grateful to Dr. Yota Otachi for his insightful comments. 

%
%
%

\end{document}